\documentclass{article}

\usepackage{xcolor}
\usepackage{todonotes}
\usepackage[top=25mm, bottom=25mm]{geometry}

\usepackage{tikz}
\usetikzlibrary{decorations.pathreplacing}
\usepackage[latin1]{inputenc}
\usepackage[english]{babel}
\usepackage{upgreek}
\usepackage{tipa}
\usepackage{lastpage}
\usepackage{float}
\usepackage[shortlabels]{enumitem}
\usepackage{subcaption}
\usepackage{multirow}
\usepackage{graphicx} 
\graphicspath{ {c:/Users/Owner/Desktop} }
\usepackage{amsmath}
\usepackage{amssymb}
\usepackage{bbm}
\usepackage{dsfont}
\usepackage{amsthm}
\usepackage{comment}
\usepackage{hyperref}
\usepackage{listings}
\usepackage{mathtools}
\usepackage{longtable}
\usepackage{lipsum}
\usepackage{abstract}
\usepackage[T1]{fontenc}

\let\oldsqrt\sqrt
\def\sqrt{\mathpalette\DHLhksqrt} \def\DHLhksqrt#1#2{%
\setbox0=\hbox{$#1\oldsqrt{#2\,}$}\dimen0=\ht0
\advance\dimen0-0.2\ht0
\setbox2=\hbox{\vrule height\ht0 depth -\dimen0}%
{\box0\lower0.4pt\box2}}

\def\b0{\boldsymbol{0}}

\newcommand{\R}     {\mathbb{R}} 
 
\newcommand{\N}     {\mathbb{N}} 
\renewcommand{\P}   {\mathbb{P}}

\newcommand{\Ecal}   {{\mathcal E }}

\newcommand{\Ical}   {{\mathcal I }}

\newcommand{\Exp}{\mathscr{E}\kern-0.2mm{\operatorname{xp}}}
\newcommand{\Log}{\mathscr{L}\kern-0.2mm{\operatorname{og}}}

\def\1{{\mathchoice {1\mskip-4mu\mathrm l}      
{1\mskip-4mu\mathrm l}
{1\mskip-4.5mu\mathrm l} {1\mskip-5mu\mathrm l}}}

\usepackage{mathtools}

\newtheoremstyle{plain}
  {6pt}
  {4pt}
  {\slshape}
  {}
  {\bfseries}
  {.}
  {0.5em}
  {}%
\newtheorem{thm}{\protect\theoremname}

  \newtheorem{prop}[thm]{\protect\propositionname}

  \newtheorem{lem}[thm]{\protect\lemmaname}
  \numberwithin{thm}{section}

\usepackage[english]{babel}
  \addto\captionsbritish{\renewcommand{\corollaryname}{Corollary}}
  \addto\captionsbritish{\renewcommand{\definitionname}{Definition}}
  \addto\captionsbritish{\renewcommand{\factname}{Fact}}
  \addto\captionsbritish{\renewcommand{\propositionname}{Proposition}}
  \addto\captionsbritish{\renewcommand{\remarkname}{Remark}}
  \addto\captionsbritish{\renewcommand{\theoremname}{Theorem}}
  \addto\captionsenglish{\renewcommand{\corollaryname}{Corollary}}
  \addto\captionsenglish{\renewcommand{\definitionname}{Definition}}
  \addto\captionsenglish{\renewcommand{\factname}{Fact}}
  \addto\captionsenglish{\renewcommand{\propositionname}{Proposition}}
  \addto\captionsenglish{\renewcommand{\remarkname}{Remark}}
  \addto\captionsenglish{\renewcommand{\theoremname}{Theorem}}
  \providecommand{\corollaryname}{Corollary}
  \providecommand{\definitionname}{Definition}
  \providecommand{\factname}{Fact}
  \providecommand{\propositionname}{Proposition}
  \providecommand{\remarkname}{Remark}
\providecommand{\theoremname}{Theorem}
\providecommand{\lemmaname}{Lemma}
 
\usepackage{parskip}
\usepackage{enumitem}

\usepackage{thmtools}
\usepackage{thm-restate}
\makeatletter
\newcommand*{\defeq}{\mathrel{\rlap{%
                     \raisebox{0.3ex}{$\m@th\cdot$}}%
                     \raisebox{-0.3ex}{$\m@th\cdot$}}%
                     =}
\makeatother

\newcommand{\julian}[1]{{\color{red} \sf $\spadesuit$ Julian: [#1]}}

\newcommand{\linknmb}{M}

\newcommand{\colnmb}{\theta}

\theoremstyle{definition}

\theoremstyle{plain}

\newtheorem{lemma}{Lemma}[section]
\newtheorem{corollary}{Corollary}[section]

\theoremstyle{remark}

\title{Improved bounds for connection probabilities in random loop models}
\author{Volker Betz, Andreas Klippel, Julian Nauth}
\date{\today}

\begin{document}

\maketitle

\begin{abstract}
We revisit and extend results by Ueltschi \cite{Ueltschi_2013} on the application of reflection positivity to loop models with $\theta \in \mathbb{N}_{\geq 2}$. By exploiting additional flexibility in the method, we prove the existence of long loops over a broader range of parameters $u$ and $\theta$, and establish new lower bounds for connection probabilities and the critical parameter~$\beta_c$. Our results are compared with recent numerical simulations, providing further insight into the phase diagram of quantum spin systems.
\end{abstract}
\section{Introduction}
In the last decades, many physical models have been investigated using probabilistic graphical representations. A class of such models are random loop models, which represent certain quantum spin systems.  For the ferromagnetic quantum Heisenberg model, such a loop
 model was introduced by Powers \cite{powers_heisenberg} and then used by T\'oth \cite{toth_pressure} to give a lower bound on the pressure. A different random loop model serves as a representation of the anti-ferromagnetic Heisenberg model, as found by Aizenman and Nachtergaele \cite{aizenman_nachtergaele}. Ueltschi \cite{Ueltschi_2013} observed that these are two instances of a one-parameter family of loop models representing XXZ-models. For the definition of this family we refer to section 2.
 
This family of loop models is built using two independent Poisson processes on the space \( E \times [0,1)_{\text{per}} \), where $E$ are nearest neighbor edges. The first process, with intensity \( u\beta \), generates crosses, while the second, with intensity \( (1 - u) \beta \), generates double bars.
Starting from a vertex \( x \) at time \( t \), we construct loops by moving upwards along the time axis until encountering either a cross or a double bar. Upon reaching a cross or double bar, we jump to the nearest right neighbor. If we encounter a double bar, we reverse direction along the time axis, whereas if we encounter a cross, we also change direction. We continue this process until we return to the starting point, forming a closed trajectory in \( V \times [0,1)_{\text{per}} \). The connection to spin-\(S\) quantum systems   requires an additional parameter $\theta$, which reweights the measure by a factor of \(\theta^{|\mathcal{L}|}\) with \(\theta = 2S + 1\). 
This introduces an additional dependency into the model, making its analysis particularly challenging.
The main question of interest is whether a phase transition occurs or not. 
For graphs with a sufficiently high vertex degree, one expects a phase transition in the sense that there exists a critical (loop) parameter \(\beta_c > 0\) such that:
\begin{enumerate}[(a)]
    \item for \(\beta < \beta_c\), only finite loops occur, and
    \item for \(\beta > \beta_c\), infinite loops appear almost surely.
\end{enumerate}
A necessary condition for (b) is the existence of an infinite percolation cluster induced by the Poisson processes. 
A remarkable result by Schramm \cite{schramm_compositions_2005} shows in the case $u=1$ and $\theta=1$ that on the complete graph \(K_n\), where edge rates are \(\beta/n\), the existence of infinite loop and the emergence of an infinite percolation cluster are asymptotically equivalent as \(n \to \infty\), implying that their critical parameters coincide. Moreover, it is shown that the normalised ordered sizes of the loops in the giant component exhibit a Poisson-Dirichlet \(\operatorname{PD}(1)\) structure.  

For graphs with uniformly bounded degree, it was shown in \cite{muehlbacher_critical_2021} that the existence of infinite percolation is not a sufficient condition. In this setting, there exists a gap between critical values, for which a lower bound is given in \cite{betz2024looppercolationversuslink}. A similar phenomenon was observed for a wide class of trees in \cite{klippel2025loopvsbernoullipercolation}.
In general to show (b) is a challenging problem. 
Nonetheless, there are some cases which are well understood.

On regular trees, Angel \cite{angel_permutations} established the existence of two distinct phases for $\theta=1$ on $d$-regular trees with $d\geq 4$ and showed that infinite loops appear for $\beta \in (d^{-1}+\tfrac{13}{6}d^{-2},\log(3))$ when $d$ is sufficiently large. Hammond \cite{hammond_infinitecycles} later proved that for $d\geq 2$, there exists a critical value $\beta_0$ above which infinite loops emerge in the case $u=1$, a result extended by Hammond and Hegde \cite{hammond_infiniteloops} to all $u\in[0,1]$. Bj\"ornberg and Ueltschi provided an asymptotic expression for the critical parameter for loops up to second order in $d^{-1}$ \cite{bjornberg_ueltschi_trees1} and further extended this analysis to the case $\theta \neq 1$ \cite{bjornberg_ueltschi_treestheta}. A comprehensive picture for $d$-regular trees with $d\geq 3$, $\theta=1$, and $u\in[0,1]$ was established by Betz et al.\ in \cite{betz_sharptrees}, proving a (locally) sharp phase transition for infinite loops in the interval $(0,d^{-1/2})$ and deriving explicit bounds on the critical parameter up to order 5 in $d^{-1}$.  

A first result for (b) on $\mathbb{Z}^d$ for $d\geq 5$ in the case of $\theta=1$ and $u=1$ was given in \cite{elboim2024infinitecyclesinterchangeprocess}.
For \(\theta \neq 1\) and $\theta \in \mathbb{N}$, there are results proving the existence of a phase transition in dimensions three and higher using the method of reflection positivity introduced by Fr\"ohlich, Simon, and Spencer for the classical Heisenberg model in \cite{frohlich1976infrared}, and introduced to the quantum case by Dyson, Lieb, and Simon \cite{dyson1978phase}. The latter corresponds to the quantum Heisenberg antiferromagnet, which is equivalent to the loop model with \(\theta = 2\) and \(u = 0\). Ueltschi managed to rigorously adapt the method to the random loop model in \cite{Ueltschi_2013} and demonstrated that it remains applicable for various values of $\theta\in \mathbb{N}_{\geq 2}$ and $u\in [0,1/2]$.

The purpose of this paper is to revisit end extend the results of \cite{Ueltschi_2013}. We use some additional flexibility that is inherent in the method, but not exploited in \cite{Ueltschi_2013}. This allows to prove long loops for a larger range of values  $\theta$ and $u$, and to obtain lower bounds on nearest-neighbor, finite-range, and long-range connections probabilities. We also give lower bounds for the critical value $\beta_c$ and compare them with numerical studies from \cite{barp2015numerical}. 

\paragraph{Outline of the remaining sections.}In Section 2, we introduce the model and some of its basic properties. Moreover, we state our main theorem and provide its proof. In Section 3, we present all the consequences of our theorem and prove them.

\section{Model and results}
The random loop model is well-established and precisely defined in various papers, such as in \cite{Ueltschi_2013} and \cite{betz_sharptrees}. However, we provide here a brief overview of the main objects.

The  model can be defined on any graph \( G \), but in this paper, we only consider the \( d \)-dimensional torus $\mathcal{G}_L=(\Lambda_L,\mathcal{E}_L) $ with side length \( L \), where \( \Lambda_L \coloneqq \mathbb{Z}^d / L \mathbb{Z}^d \) is the set of vertices, and the edge set $\mathcal{E}_L$ consists of nearest neighbors, with periodic boundary conditions.

Let \( u \in [0,1] \) and $\beta > 0$. We attach the interval $[0,\beta)$ to each edge of the graph, and will interpret $t \in [0,\beta)$ as a time parameter. On the space \( \mathcal{E}_L \times [0, \beta) \), we define two independent Poisson point processes: one with intensity \( u \) that produces so-called crosses, and another with intensity \( 1 - u \) that produces double bars. A link is either a cross or a double bar. The combined configuration space is denoted by \( \Omega(G)\), with probability measure \( \rho_{G,\beta,u} \). For a realization \( \omega \in \Omega(G) \), we denote by \( \linknmb \) the number of links.

The crosses and double bars form loops, defined by the following rules:

\begin{itemize}
    \item We start at a point \( (x,t)\in V\times[0,\beta) \) and move in the positive time direction, thereby exploring links on edges adjacent to $x$.
    \item If we reach time $\tau$ and encounter a double bar at $(\{x,y\}, \tau-)$, we jump to $(y,\tau)$ and change time direction.
    \item If we reach time $\tau$ and encounter a cross at $(\{x,y\}, \tau-)$, we jump to $(y,\tau)$ and continue in the same time direction.
    \item If, for any vertex $z$, we reach $(z,\beta)$ whilst moving in the positive time direction then we jump to $(z,0)$ and continue moving in the positive time direction. If we reach $(z,0)$ whilst moving in the negative time direction then we jump to $(z,\beta)$ and continue moving in the negative time direction.
    \item If we reach $(x,t)$ again, the loop is complete.
\end{itemize}

We observe that loops are closed trajectories in the space \( \Lambda_L\times [0, \beta) \).

We denote the set of loops for a given configuration \( \omega \) by \( \mathcal{L}(\omega) \) and fix \( \theta \in \mathbb{N}_{>0} \). We then color these loops uniformly with color $i \in [\theta]$ and reweight the measure by the number of colors. The loop-weighted measure is defined by
\begin{equation*}\label{eq:prob_meas_rlm}
    \mathbb{P}_{G,\beta,u,\theta}(\omega) \coloneqq \frac{1}{\widetilde{Z}} \theta^{|\mathcal{L}(\omega)|} \rho_{G,\beta,u},
\end{equation*}
where \( \widetilde{Z} \) is the partition function of the random loop model, given by
\begin{equation*}
    \widetilde{Z} \coloneqq \int_{\Omega(G)} \rho_{G,\beta,u} \, \theta^{|\mathcal{L}(\omega)|} \, \mathrm{d}\omega.
\end{equation*}
We denote the measure where the underlying graph is the torus with side-length $L$ by
$$\mathbb{P}_{L,\beta,u,\theta}.$$
Note that this model is typically defined more generally for \( \theta \in \mathbb{R}_{>0} \). Since we are only interested in the probabilistic interpretation, we restrict ourselves to natural numbers and set the color set as \( [\theta] = \{1, \ldots, \theta\} \). For an introduction that includes \( \theta \in \mathbb{R}_{>0} \), we recommend \cite{Ueltschi_2013}.

Based on this framework, we define the concept of connection within the model. Two points \( (x,0) \) and \( (y,t) \) are said to be connected through time $t$, denoted by $E_{x,y,t}$, if there exists a loop connecting them through time $t$ that adheres to the rules outlined above. This connection indicates that the two points belong to the same closed trajectory within the random loop configuration.
Furthermore, we denote the infinite volume measure by 
\[
\P_{\beta,u,\theta} \coloneqq \lim_{L \to \infty} \P_{L,\beta,u,\theta}.
\]
The existence of these limits has been discussed in \cite{froelich_vol}. 

For the statement of Theorem \ref{t:main} below, let 
$c=(c_\ell)_{0\leq \ell \leq m}\in\R^{m+1}$ with $m\in\N$, $u \in [0,1]$ and $d \in \N$. We define 
\begin{align} \label{eq: curly I}
\Ical^{u,d}_{c}(\alpha)\coloneqq&
\int_{[0,2\pi]^d}\frac{\mathrm{d}^dk}{(2\pi)^d}\, \sqrt{u\alpha + (1-u)(1-\alpha)\frac{\epsilon(k+\pi)}{\epsilon(k)}} \bigg( \sum_{\ell=0}^{m} \sum_{j=1}^d \frac{c_\ell}{d} \cos( \ell k_j) \bigg)_+
\,,
\end{align}
with $k+\pi\coloneqq(k_j+\pi)_{j\in\{1,\ldots,d\}}$, 
\[
\epsilon(k)\coloneqq 2\sum_{j=1}^d \big(1-\cos (k_j)\big)
\]
for $k\in[0,2\pi]^d$, and $x_+\coloneqq \max\{0,x\}$ for  $x\in\mathbb{R}$. In addition, let 
\begin{equation}\label{eq:Idc_def}
I^{u,d}_{c} \coloneqq \sup_{\alpha\in[0,1]} \Ical^{u,d}_{c}(\alpha) \quad \text{and} \quad 
\tilde I^d_c \coloneqq
\int_{[0,2\pi]^d}\frac{\mathrm{d}^dk}{(2\pi)^d}\, \frac{1}{\epsilon(k)} \bigg( \sum_{\ell=0}^{m} \sum_{j=1}^d \frac{c_\ell}{d} \cos (\ell k_j) \bigg)_+.
\end{equation}

Note that $\tilde I_c^d < \infty$ if and only if  
either $\sum_{\ell=0}^{m} c_\ell = 0$ or $d \geq 3$.

Our main result reads as follows:
\begin{thm}\label{t:main}
Let $d\geq 1$, $\theta\geq 2$, $u\in[0,1/2]$, $\beta>0$ and $e\in\Lambda_L$ be a unit vector in the lattice geometry. Let further $m\in\N$ and $c = (c_0,\dots,c_{m-1},c_{m}) \in\mathbb{R}^{m+1}$. Then,

\begin{align*}
\bigg(\sum_{\ell=0}^{m} c_\ell\bigg) \lim_{|\Lambda|\rightarrow\infty} \bigg( \frac{1}{|\Lambda|} \sum_{x\in\Lambda} \P_{\beta,u,\theta}(E_{0,x,0}) \bigg)
&\geq
\sum_{\ell=0}^{m} c_\ell \P_{\beta,\theta,u}\big(E_{0, \ell e,0}\big) \\
&\quad - \colnmb\sqrt{\frac{\P_{\beta,\theta,u}\big(E_{0,e,0}\big)}{2}} I^{u,d}_{c}
- \frac{\colnmb}{2\beta} \tilde I^d_c
\,.
\end{align*}

\end{thm}

The result of \cite{Ueltschi_2013} covers the special cases $c = (1,0)$ and $c = (0,1)$ above. In these cases,
the supremum involved in the definition of $I_{c}^{u,d}$ is taken at $\alpha=1$. We will discuss this below, 
but first we give a proof of our theorem that starts from the relevant estimates given in \cite{Ueltschi_2013} and 
only adds those parts that need to be added and modified. 

\begin{proof}[Proof of Theorem \ref{t:main}]

We work with the Fourier transform of the connection probabilities. Let us therefore introduce the notation 

$$\kappa(x,t)=\P_{\beta,u,\theta}(E_{0,x,t}).$$
Then the Fourier transform is given by
\begin{align*}
&\hat{\kappa}(k,t) \coloneqq \sum_{x\in\Lambda} e^{-ik\cdot x}  \kappa(x,t), \quad
\tilde{\kappa}(x,\tau)\coloneqq \int_0^\beta\text{d}t\, e^{-i\tau t}  \kappa(x,t), \\
&\hat{\tilde{\kappa}}(k,\tau)\coloneqq \sum_{x\in\Lambda} \int_0^\beta\text{d}t\, e^{-i\tau t-ik\cdot x}  \kappa(x,t),
\end{align*}

where $k$ belongs to the set 
$$\Lambda^*=\bigg\{k\in \tfrac{2\pi}{L}\mathbb{Z}^d: -\pi<k_i\leq\pi,\, i=1,\dots,d\bigg\}, $$
and $\tau \in \tfrac{2\pi}{\beta}\mathbb{Z}^d$. 
Since $ \kappa(x,t)$ is symmetric w.r.t. $x$ and $t$, the complex oscillation $e^{-i.}$ can equally be replaced by a cosine.\\
Let $k\in \Lambda^*\setminus \{0\}$. We will use the following two inequalities from 
\cite{Ueltschi_2013}:
The first inequality in \cite[equation (3.24)]{Ueltschi_2013} yields $\hat{\kappa}(k,0)\geq0$. The necessary argument needs quantum spin notation and follows the one in \cite{Ueltschi_2013} rather closely, which is why we don't give it here. For the convenience of the reader, it is detailed in the appendix. Inequality (5.54) of the same reference states that for $u\in [0,\tfrac{1}{2}]$ and $\theta=2,3,\dots$, we have

\begin{equation}\label{Falkbruch+} \hat{\kappa}(k,0)\leq \colnmb\sqrt{\tfrac{\P_{\beta,u,\theta}(E_{0,e,0})}{2\epsilon(k)}}\sqrt{ u\epsilon(k)\alpha' + (1-u)(1-\alpha')\epsilon(k+\pi)} + \frac{1}{\beta} \frac{\colnmb}{2\epsilon(k)} , \end{equation}

   where $\alpha' \in [0,1]$ is obtained from a conditional probability introduced in \cite[equality (5.53)]{Ueltschi_2013}. The value of $\alpha'$ is not known, which is why in the statement of Theorem 
   \ref{t:main} the supremum over $\alpha$ appears.

By definition
\begin{align*}
|\Lambda| \kappa(\ell e_j,0)
=& \sum_{k\in \Lambda^*} e^{-i \ell e_j\cdot k} \hat{\kappa}(k,0)
\\=&\, \hat{\kappa}(0,0) + \sum_{\substack{
k\in \Lambda^*\setminus\{0\}
}} e^{-i \ell k_j} \hat{\kappa}(k,0)
\\=& \sum_{x\in\Lambda}  \kappa(x,0) + \sum_{\substack{
k\in \Lambda^*\setminus\{0\}
}} \hat{\kappa}(k,0) \cos(\ell k_j) \,.
\end{align*}
Multiplication with $c_\ell$ and summation over $\ell$ gives
\begin{align}\label{eq:corr_dens_0}
\left(\sum_{m=0}^{m} c_\ell\right) \frac{1}{|\Lambda|} \sum_{x\in\Lambda}  \kappa(x,0)
&= \sum_{\ell=0}^{m} c_\ell  \kappa(\ell e_j,0) - \frac{1}{|\Lambda|}\sum_{\substack{
k\in \Lambda^*\setminus\{0\}
}} \hat{\kappa}(k,0) \sum_{\ell =0}^{m} c_m\cos(\ell k_j)
\,.
\end{align}
Summation over $j$ and multiplying with $1/d$ leads to
\begin{align}
\left(\sum_{\ell=0}^{m} c_\ell \right) \frac{1}{|\Lambda|} \sum_{x\in\Lambda} \kappa(x,0)
=& \sum_{\ell =0}^{m} c_\ell \bigg(\frac{1}{d}\sum_{j=1}^d  \kappa(\ell e_j,0)\bigg)
\notag\\ 
&- \frac{1}{|\Lambda|} \frac{1}{d}\sum_{j=1}^d\sum_{\substack{
k\in \Lambda^*\setminus\{0\}
}} \hat{\kappa}(k,0) \sum_{\ell =0}^{m} c_\ell \cos(\ell k_j)
\notag\\ 
=& \sum_{\ell =0}^{m} c_\ell \kappa(\ell e,0) - \frac{1}{|\Lambda|}\sum_{\substack{
k\in \Lambda^*\setminus\{0\}
}} \hat{\kappa}(k,0) \sum_{m=0}^{N-1} \sum_{j=1}^d \frac{c_\ell}{d}\cos(\ell k_j),
\end{align}
where we used the lattice symmetry for $ \kappa(\ell e_j,0)$.
Since $\hat\kappa(k,0)\geq 0$, which is proven in the Appendix, taking the positive part yields an upper bound

\begin{align}\label{eq:corr_dens_1}
\left(\sum_{\ell =0}^{m} c_\ell \right) \frac{1}{|\Lambda|} \sum_{x\in\Lambda} \kappa(x,0)
\geq& \sum_{m=0}^{N-1} c_\ell \kappa(\ell e,0) - \frac{1}{|\Lambda|}\sum_{\substack{
k\in \Lambda^*\setminus\{0\}
}} \hat{\kappa}(k,0) \left( \sum_{\ell=0}^{m} \sum_{j=1}^d \frac{c_\ell}{d} \cos (\ell k_j) \right)_+ \,.
\end{align}

Now we use (\ref{Falkbruch+}) to obtain

\begin{align*}
&\left(\sum_{\ell=0}^{m} c_\ell\right) \frac{1}{|\Lambda|} \sum_{x\in\Lambda}  \kappa(x,0)
\notag\\\geq& \sum_{\ell=0}^{m} c_\ell \kappa(\ell e,0)
- \frac{\colnmb}{2\beta|\Lambda|}\sum_{\substack{
k\in \Lambda^*\setminus\{0\}
}} \frac{1}{\epsilon(k)} \left( \sum_{\ell =0}^{m} \sum_{j=1}^d \frac{c_\ell}{d} \cos (\ell k_j) \right)_+
\notag\\& - \frac{\colnmb}{|\Lambda|}\sqrt{\frac{p_{e,0}}{2}}
\sum_{\substack{
k\in \Lambda^*\setminus\{0\}
}}\sqrt{ u\alpha' + (1-u)(1-\alpha')\frac{\epsilon(k+\pi)}{\epsilon(k)}} \left( \sum_{\ell=0}^{m} \sum_{j=1}^d \frac{c_\ell}{d} \cos (\ell k_j) \right)_+ .
\end{align*}
Taking the limit $|\Lambda|\rightarrow\infty$ and estimating $\alpha'$ by the supremum over $\alpha\in[0,1]$ reveals the integral terms $I^{u,d}_{c}$ and $\tilde I^{d}_{c}$, which concludes the proof.
\end{proof}

We end this section with a discussion of the quantity $\Ical^{u,d}_{c}(\alpha)$ given in \eqref{eq: curly I}. 
We first by recall a tool introduced in \cite{Ueltschi_2013} that can be applied for 
relevant choices of $c$. 
\begin{lem}\label{l:I_max_der}
Let $d\geq 1$, $u\in[0,1/2]$, $m\in\N$, and $c\in\mathbb{R}^{m+1}$.
\begin{itemize}
\item If $(\Ical^{u,d}_{c})'(0)\leq 0$, then $\Ical^{u,d}_{c}(\alpha)$ takes its supremum at $\alpha=0$.
\item If $(\Ical^{u,d}_{c})'(1)\geq 0$, then $\Ical^{u,d}_{c}(\alpha)$ takes its supremum at $\alpha=1$.
\end{itemize}
\end{lem}
\begin{proof}
We consider the function
\begin{align*}
f_k(\alpha) \coloneqq \sqrt{ u\alpha + (1-u)(1-\alpha)\tfrac{\epsilon(k+\pi)}{\epsilon(k)}} \,,
\end{align*}
covering the dependence on $\alpha$ of $\Ical^{u,d}_{c}$ inside the integral. Its derivatives read as
\begin{align*}
f_k'(\alpha) &= \frac{
u\epsilon(k) - (1-u)\epsilon(k+\pi)}{2\sqrt{\epsilon(k)}\sqrt{
u\alpha\epsilon(k) + (1-u)(1-\alpha)\epsilon(k+\pi)}}\,,
\\
f_k''(\alpha) &= -\frac{
(u\epsilon(k) - (1-u)\epsilon(k+\pi))^2}{4\sqrt{\epsilon(k)}\sqrt{
u\alpha\epsilon(k) + (1-u)(1-\alpha)\epsilon(k+\pi)}^3} \,.
\end{align*}
$\Ical^{u,d}_{c}$ is concave since $f''(\alpha)\leq 0$ for all $\alpha\in[0,1]$. The claims follow from standard techniques using the mean value theorem. 
\end{proof}
This tool allows to compute the case $u=0$ straightforwardly. We set 
\begin{align}\label{eq:Jdc_def}
    J^d_c \coloneqq \Ical_{c}^{0,d}(0) = 
    \int_{[0,2\pi]^d} \frac{\mathrm{d}^dk}{(2\pi)^d}\, \sqrt{\frac{\epsilon(k + \pi)}{\epsilon(k)}} \left( \sum_{\ell=0}^{m} \sum_{j=1}^d \frac{c_\ell}{d} \cos(\ell k_j) \right)_+ \,.
\end{align}

\begin{prop}\label{l:I_max_u=0}
Let $d\geq 1$, $m\in\N$, and $c\in\mathbb{R}^{m+1}$. If $u=0$, then $I^{u,d}_{c} = J^d_c$ with $J^d_c$ defined in Eq. \eqref{eq:Jdc_def}.
\end{prop}
\begin{proof}
In the case of $u=0$, the derivative of $\Ical^{u,d}_{c}(\alpha)$ simplifies to
\begin{align*}
(\Ical^{0,d}_{c})'(\alpha)=
 \int_{[0,2\pi]^d}\frac{\mathrm{d}^dk}{(2\pi)^d}\, \frac{
-\epsilon(k+\pi)}{2\sqrt{
(1-\alpha)\epsilon(k)\epsilon(k+\pi)}} \bigg( \sum_{\ell =0}^{m} \sum_{j=1}^d \frac{c_\ell}{d} \cos(\ell k_j) \bigg)_+ \leq 0
\end{align*}
for all $\alpha\in[0,1]$. Applying Lemma \ref{l:I_max_der} and computing $\Ical^{0,d}_{c}(0)$ yields the claim.
\end{proof}
For general $u\in[0,1/2]$, this statement equally holds true under certain conditions for the coefficients $c$:
\begin{prop}\label{l:I_max_m_even}
Let $d\geq 1$, $u\in[0,1/2]$, and $m\in\N$. Let further $c\in\mathbb{R}^{m+1}$ fulfill $c_1\geq 0$ and $c_\ell=0$ for all odd $\ell \geq 3$. Then, $I^{u,d}_{c} = \Ical_{c}^{u,d}(0) =  \sqrt{1-u} J^d_c$.
\end{prop}
\begin{proof}
In keeping with Lemma \ref{l:I_max_der}, we consider
\begin{align*}
(\Ical^{u,d}_{c})'(0)
=&  \int_{[0,2\pi]^d}\frac{\mathrm{d}^dk}{(2\pi)^d}\, \frac{
u\epsilon(k) - (1-u)\epsilon(k+\pi)}{2\sqrt{
(1-u)\epsilon(k)\epsilon(k+\pi)}} \bigg( \sum_{\ell =0}^{m} \sum_{j=1}^d \frac{c_\ell}{d} \cos(\ell k_j) \bigg)_+
\notag\\\leq&
\frac{1}{2\sqrt{2(1-u)}} \int_{[0,2\pi]^d}\frac{\mathrm{d}^dk}{(2\pi)^d}\, \frac{
\epsilon(k) - \epsilon(k+\pi)}{\sqrt{\epsilon(k)\epsilon(k+\pi)}} \bigg( \sum_{\ell =0}^{m} \sum_{j=1}^d \frac{c_\ell}{d} \cos(\ell k_j) \bigg)_+
\end{align*}
using that $u\leq 1/2$ and thus $1-u\geq 1/2$. Next, we abbreviate $g(k) \coloneqq \sum_{j=1}^d\cos(k_j)/d$. Noting that $\epsilon(k) = 2 d\big(1-g(k)\big)$ and $\epsilon(k+\pi) = 2 d\big(1+g(k)\big)$ yields
\begin{align}\label{eq:I'(0)_m_even}
(\Ical^{u,d}_{c})'(0) \leq&
\frac{1}{\sqrt{2(1-u)}} \int_{[0,2\pi]^d}\frac{\mathrm{d}^dk}{(2\pi)^d} \, \frac{
-g(k)}{\sqrt{\big(1-g(k)\big)\big(1+g(k)\big)}} 
\notag\\&\times \bigg( c_1g(k) + \sum_{\substack{\ell=0 \\ \ell \text{ even}}}^m \sum_{j=1}^d \frac{c_\ell}{d} \cos(\ell k_j) \bigg)_+
\,,
\end{align}
We split the integral into the cases $g(k)>0$ and $g(k)<0$ while the integral restricted to $g(k)=0$ vanishes. In the second case, we substitute $k\mapsto k+\pi$ to give
\begin{align*}
&\int_{[0,2\pi]^d}\mathrm{d}^dk\,\textbf{1}_{g(k)<0} \frac{
-g(k)}{\sqrt{1-g(k)^2}} \bigg(c_1g(k) + \sum_{\substack{\ell=0 \\ \ell \text{ even}}}^m \sum_{j=1}^d \frac{c_\ell }{d} \cos(\ell k_j) \bigg)_+
\notag\\=&
\int_{[0,2\pi]^d}\mathrm{d}^dk\,\textbf{1}_{g(k)>0} \frac{
g(k)}{\sqrt{1-g(k)^2}} \bigg(-c_1g(k) + \sum_{\substack{\ell =0 \\ \ell \text{ even}}}^m \sum_{j=1}^d \frac{c_\ell}{d} \cos(\ell k_j) \bigg)_+
\notag\\\leq &
\int_{[0,2\pi]^d}\mathrm{d}^dk\,\textbf{1}_{g(k)>0} \frac{
g(k)}{\sqrt{1-g(k)^2}} \bigg(c_1g(k) + \sum_{\substack{\ell=0 \\ \ell \text{ even}}}^m \sum_{j=1}^d \frac{c_\ell}{d} \cos(\ell k_j) \bigg)_+
\end{align*}
using that $c_1g(k)\geq 0$ in the last step. We notice that the latter expression corresponds to the integral in Eq. (\ref{eq:I'(0)_m_even}) restricted to $g(k)>0$ except for the reversed sign. Inserting this estimation into Eq. (\ref{eq:I'(0)_m_even}) thus reveals $(\Ical^{u,d}_{c})'(0) \leq 0$. Lemma \ref{l:I_max_der} concludes the proof.
\end{proof}
We now discuss the case $c = (1,-1)$. While the case $u=0$ is already covered by Lemma \ref{l:I_max_u=0}, the case $u=1/2$ is discussed in the following:
\begin{prop}\label{l:I_max_(1,-1)}
Let $d\geq 1$ and set $c=(1,-1)$. If $u=1/2$, then $I^{u,d}_{c}=1/\sqrt{2}$.
\end{prop}
\begin{proof}
According to Lemma \ref{l:I_max_der}, it suffices to show that
\begin{align*}
(\Ical^{1/2,d}_{(1,-1)})'(1) =&  \int_{[0,2\pi]^d}\frac{\mathrm{d}^dk}{(2\pi)^d}\, \frac{
\epsilon(k)-\epsilon(k+\pi)}{4\sqrt{
2}\epsilon(k)} \bigg(1-\frac{1}{d}\sum_{j=1}^d \cos(k_j)\bigg)_+
\notag\\=& \frac{1}{8\sqrt{
2}d}  \int_{[0,2\pi]^d}\frac{\mathrm{d}^dk}{(2\pi)^d}\,
\big(\epsilon(k)-\epsilon(k+\pi)\big) 
\notag\\=&\, 0 \,,
\end{align*}
where the last step follows from substituting $k\mapsto k+\pi$ in the second term. The claim follows as
\begin{align*}
I^{1/2,d}_{(1,-1)} = \Ical^{1/2,d}_{(1,-1)}(1)
=  \int_{[0,2\pi]^d}\frac{\mathrm{d}^dk}{(2\pi)^d}\, \frac{1}{\sqrt{2}} \bigg(1- \frac{1}{d} \sum_{j=1}^d \cos(k_j) \bigg)_+
= \frac{1}{\sqrt{2}} \,.
\end{align*}
\end{proof}

\section{Consequences}
\subsection{Lower bound for nearest neighbor connection probabilities}

In the following, we vary \( c \) and derive several corollaries. First, we establish a lower bound on the connection probability \( \P_{\beta,\theta,u}\big(E_{0,e,0}\big) \). To this end, we define
\begin{align*}
    P_{\beta,\theta,u,d} \coloneqq
    \Bigg(\sqrt{\bigg(\frac{\colnmb}{2\sqrt{2}} I^{u,d}_{(1,-1)}\bigg)^2 + 1 - \frac{\colnmb}{4d\beta}} - \frac{\colnmb}{2\sqrt{2}} I^{u,d}_{(1,-1)}\Bigg)^2
\end{align*}
for all values of $\beta>0$ such that $1 - \frac{\theta}{4 d \beta} \geq 0$. We define  $P_{\theta,u,d}\coloneqq\lim_{\beta\to\infty}P_{\beta,\theta,u,d}$.

\begin{prop}\label{c:m=1_bound}
    Let \( d \geq 1 \) and $\beta \geq \frac{\theta}{4d}$. Then we have
    \[
        \P_{\beta,\theta,u}\big(E_{0,e,0}\big) \geq P_{\beta,\theta,u,d}.
    \]
\end{prop}

\begin{proof}
    Take \( N =2, c = (1, -1) \) in Theorem \ref{t:main}. Then we obtain
    \begin{align}\label{eq:abs_bound}
        0 \geq 1-\P_{\beta,\theta,u}\big(E_{0,e,0}\big) - \colnmb \sqrt{\tfrac{\P_{\beta,\theta,u}\big(E_{0,e,0}\big)}{2}} I^{u,d}_{(1,-1)} - \frac{\colnmb}{4d\beta} \,,
    \end{align}
    where we computed $\tilde I^d_{(1,-1)} = 1/(2d)$. By solving the associated quadratic equation in $\sqrt{\P_{\beta,\theta,u}(E_{0,e,0})}$, we derive the desired bound.
\end{proof}
For the special case $u=1/2$, Proposition \ref{l:I_max_(1,-1)} leads to the simpler expression 
   \begin{equation*}
       P_{\beta,\theta,1/2,d} =
\bigg(\sqrt{\Big(\frac{\colnmb}{4}\Big)^2+1- \frac{\colnmb}{4d\beta}}-\frac{\colnmb}{4} \bigg)^2 \,.
   \end{equation*}
For $u=0$, we can use Proposition \ref{l:I_max_u=0} in order to compute $ P_{\beta,\theta,1/2,d}$ numerically. 
Tables \ref{tab:p_bounds_m=1_u=0} and \ref{tab:p_bounds_m=1_u=1/2} contain results of these calculations for $\beta \to \infty$.

\begin{table}[!ht]
\centering
{\renewcommand{\arraystretch}{1.2}
\begin{tabular}{cccccccccc}
\hline\hline
$\theta$ & $d=1$ & $d=2$ & $d=3$ & $d=4$ & $d=5$ & $d=6$ & $d=7$ & $d=8$ & $d=9$
\\ \hline
$2$ & $0.417$ & $0.323$ & $0.300$ & $0.290$ & $0.285$ & $0.282$ & $0.280$ & $0.278$ & $0.277$
\\ \hline
$3$ & $0.282$ & $0.200$ & $0.182$ & $0.174$ & $0.170$ & $0.168$ & $0.166$ & $0.165$ & $0.164$
\\ \hline
$4$ & $0.198$ & $0.132$ & $0.119$ & $0.113$ & $0.110$ & $0.108$ & $0.107$ & $0.106$ & $0.105$
\\ \hline
$5$ & $0.144$ & $0.092$ & $0.082$ & $0.078$ & $0.076$ & $0.075$ & $0.074$ & $0.073$ & $0.072$
\\\hline\hline
\end{tabular}
}
\caption{Lower bound $P_{\theta,u,d}$ of $\P_{\beta,\theta,u}\big(E_{0,e,0}\big)$ for $\beta\to\infty$, $u=0$, various dimensions $d$, and values of $\colnmb\in\N_{\geq2}$.}
\label{tab:p_bounds_m=1_u=0}
\end{table}
\begin{table}[!ht]
\centering
{\renewcommand{\arraystretch}{1.2}
\begin{tabular}{cccc}
\hline\hline
$\colnmb=2$ & $\colnmb=3$ & $\colnmb=4$ & $\colnmb=5$
\\ \hline
0.381 & 0.250 & 0.171 & 0.123
\\\hline\hline
\end{tabular}
}
\caption{Lower bound $P_{\theta,u,d}$ of $\P_{\beta,\theta,u}\big(E_{0,e,0}\big)$ for $\beta\to\infty$, $u=1/2$, and various values of $\theta\in\N_{\geq2}$. The lower bound turns out to be constant w.r.t. $d$.}
\label{tab:p_bounds_m=1_u=1/2}
\end{table}
\subsection{Lower bound for finite range connection probabilities}
Next, we derive lower bounds for \( \P_{\beta,\theta,u}\big(E_{0,me,0}\big) \) with \( m \geq 2 \). To this end, let $p\in[0,1]$ and \( h_{p} : \mathbb{R}_{\geq 0} \to \mathbb{R}_{\geq 0} \) be defined as
\begin{align}\label{eq:h_fct}
    h_{p}(x) \coloneqq 
    \begin{cases}
        1 - p + 2x\sqrt{p}, & \text{if } x \leq \sqrt{p}, \\
        1 + x^2, & \text{if } \sqrt{p} < x \leq 1, \\
        2x, & \text{if } x > 1.
    \end{cases}
\end{align}
Also, let $\eta \geq 0$, $m \geq 2$, and set 
\[
c(\eta) = (1 - \eta, \eta, 0, \ldots, 0, -1) \in \R^{m+1}. 
\]

\begin{prop}\label{c:m>=2_bound}
    For $d\geq1$ and $m\geq2$, we have
    \begin{align}\label{eq:m>=2_bound}
        \P_{\beta,\theta,u}\big(E_{0,me,0}\big) \geq \sup_{\eta \geq 0} \bigg( 1 - \eta \, h_{P_{\beta,\theta,u,d}}\left( \frac{\colnmb}{2\sqrt{2}\eta} I^{u,d}_{c(\eta)} \right) - \frac{\colnmb}{2\beta} \tilde I^d_{c(\eta)} \bigg).
    \end{align}
 
\end{prop}

\begin{proof}
Since the coefficients of $c(\eta)$ sum to zero, the left hand side in Theorem \ref{t:main} vanishes. Rearranging the remaining terms gives the inequality 
\begin{align}\label{eq:corr_bound_(1-c1,c1)}
\P_{\beta,\theta,u}\big(E_{0,me,0}\big) \geq 
(1-\eta) + \eta \P_{\beta,\theta,u}\big(E_{0,e,0}\big) - \colnmb\sqrt{\frac{\P_{\beta,\theta,u}\big(E_{0,e,0}\big)}{2}} I^{u,d}_{c(\eta)} - \frac{\colnmb}{2\beta} \tilde I^d_{c(\eta)} \,,
\end{align}
which holds for all $\eta \geq 0$. By Proposition 
\ref{c:m=1_bound}, we know that  $\P_{\beta,\theta,u}\big(E_{0,e,0}\big) \in [P_{\beta,\theta,u,d},1]$, and so we minimize the above expression over this interval. We define
\[
g(\eta) = \frac{\colnmb}{2\sqrt{2}\eta}I^{u,d}_{c(\eta)}
\]
and set 
\[
f_{\eta}(p) =  (1-\eta) + \eta p - 2 \eta \sqrt{p} g(\eta) - \frac{\colnmb}{2\beta} \tilde I^d_{c(\eta)}
= \eta \big(\sqrt{p}-g(\eta)\big)^2 - \eta g(\eta)^2 + 1 - \frac{\colnmb}{2\beta} \tilde I^d_{c(\eta)} \, .
\]
for $p\in[P_{\beta,\theta,u,d},1]$. The function $f_{\eta}$ takes its minimum at 
\[
p_\ast(\eta) = \begin{cases} 
P_{\beta,\theta,u,d} & \text{if } g(\eta) \leq \sqrt{P_{\beta,\theta,u,d}}, \\
g(\eta)^2 & \text{if } \sqrt{P_{\beta,\theta,u,d}} \leq 
g(\eta) \leq 1 \\
1 & \text{if } g(\eta) > 1,
\end{cases} .
\]
This gives 
\[
\P_{\beta,\theta,u}\big(E_{0,me,0}\big) \geq f_\eta(p_\ast(\eta)) = 1 - 
\eta h_{P_{\beta,\theta,u,d}}(g(\eta))  - \frac{\colnmb}{2\beta} \tilde I^d_{c(\eta)} \,.
\]
The result follows by maximizing over $\eta$. 
\end{proof}

In Table \ref{tab:p_bounds_m>=2} , we present lower bounds for the case $\beta\to\infty$, \( \theta = 2 \), and \( u = \frac{1}{2} \) across several dimensions and values for $m\geq 2$. For even $m$, we used Proposition \ref{l:I_max_m_even} to facilitate the numerical computation. In contrast, numerical values for $m\in\{3,5\}$ were obtained by computing the supremum over $\alpha$ numerically. 

\subsection{Lower bounds for long range order}

The next case we want to address is the limit as \( m 
\to \infty \) along with $m\in 2\N$. As a preparation, 
we investigate the limit of the quantities $J_c^d$ from  
\eqref{eq:Jdc_def} and $\tilde I_c^d$ from \eqref{eq:Idc_def} for certain sequences of vectors $c$. 

\begin{lem} \label{lem:limit}
Let $n \in \N$, $c = (c_0, \ldots, c_{n}) \in \R^{n+1}$, 
and 
$c_\infty = \sum_{\ell=0}^n c_\ell$. 
For 
$N \in \N$ with $N > n+1$ set 
\[
c(N) = (c_0, \ldots, c_n,0, \ldots, 0, c_\infty) \in 
\R^N, \quad J^d_{c,N} = J^d_{c(N)}, \quad \tilde I^d_{c,N} = \tilde I^d_{c(N)}.
\]
Then 
\[
J^d_{c,\infty} \coloneqq \lim_{N \to \infty} J^d_{c,N} = 
\int_{[0,2\pi]^{2d}}\frac{\text{d}^d k}{(2\pi)^{d}} \, \frac{\text{d}^d\tilde k}{(2\pi)^{d}}\,\sqrt{\frac{\epsilon( k+\pi)}{\epsilon( k)}}
\bigg( \sum_{j=1}^d \Big(\sum_{\ell=0}^{n}\frac{c_\ell}{d} \cos(\ell k_j) + \frac{c_\infty}{d} \cos(\tilde  k_j) \Big) \bigg)_+
\,,
\]
and 
\[
\tilde I^d_{c,\infty} \coloneqq  \lim_{N \to \infty} \tilde I^d_{c,N} = 
\int_{[0,2\pi]^{2d}}\frac{\text{d}^dk}{(2\pi)^{d}} \, \frac{\text{d}^d\tilde k}{(2\pi)^{d}}\,\frac{1}{\epsilon(k)}
\bigg( \sum_{j=1}^d \Big(\sum_{\ell=0}^{n}\frac{c_\ell}{d} \cos(\ell k_j) + \frac{c_\infty}{d} \cos(\tilde k_j) \Big) \bigg)_+ \,.
\]
\end{lem}
\begin{proof}
We have 
\begin{align*}
J^d_{c(N)}
=&  \int_{[0,2\pi]^d}\frac{\text{d}^d\tilde k}{(2\pi)^d}\, \sqrt{\frac{\epsilon(\tilde k+\pi)}{\epsilon(\tilde k)}} \bigg( \sum_{j=1}^d \Big( \sum_{\ell=0}^{n} \frac{c_\ell}{d} \cos( \ell \tilde k_j) + \frac{c_\infty}{d} \cos(N \tilde k_j) \Big) \bigg)_+
\\=&  \int_{[0,2\pi N]^d}\frac{\text{d}^d\tilde k}{(2\pi N)^d}\, \sqrt{\frac{\epsilon\big(\frac{\tilde k}{N}+\pi\big)}{\epsilon\big(\frac{\tilde k}{N}\big)}} 
\bigg( \sum_{j=1}^d\Big( \sum_{\ell=0}^{n} \frac{c_\ell}{d} \cos\Big(\ell \frac{\tilde k_j}{N}\Big) + \frac{c_\infty}{d} \cos(\tilde k_j) \Big) \bigg)_+
%
%
\\=& \sum_{ k\in\{0,\ldots,N-1\}^d}
\int_{[0,2\pi]^d+2\pi k}\frac{\text{d}^d\tilde k}{(2\pi N)^d}\,
\sqrt{\frac{\epsilon\big(\frac{\tilde k}{N}+\pi\big)}{\epsilon\big(\frac{\tilde k}{N}\big)}}
\bigg( \sum_{j=1}^d \Big( \sum_{\ell=0}^{n}\frac{c_\ell}{d} \cos\Big(\ell \frac{\tilde k_j}{N}\Big) +  \frac{c_\infty}{d} \cos(\tilde k_j) \Big) \bigg)_+
\\=&  \sum_{ k\in\{0,\ldots,N-1\}^d} \int_{[0,2\pi]^d} \frac{\text{d}^d\tilde k}{(2\pi N)^d} \, \sqrt{\frac{\epsilon\big(\frac{\tilde k+2\pi k}{N}+\pi\big)}{\epsilon\big(\frac{\tilde k+2\pi k}{N}\big)}}
\bigg( \sum_{j=1}^d \Big( \sum_{\ell=0}^{n}\frac{c_\ell}{d} \cos\Big(\ell\frac{\tilde k_j+2\pi k_j}{N}\Big) +  \frac{c_\infty}{d} \cos(\tilde k_j) \Big) \bigg)_+
\,.
\end{align*}
For $N\rightarrow\infty$, the term $\tilde k/N$ vanishes as $\tilde k\in[0,2\pi]^d$ is bounded. It should be kept in mind that the integrand does not exhibit any singularity owing to $\sum_{\ell=0}^{n}c_\ell + c_\infty=0$. The result follows from a Riemann integral as
\begin{align}\label{eq:J_lim}
J^d_{c(N)}
=& \,\, \frac{1}{N^d} \sum_{ k\in\{0,\ldots,N-1\}^d} \int_{[0,2\pi]^d} \frac{\text{d}^d\tilde k}{(2\pi)^d} \, \sqrt{\frac{\epsilon\big(\frac{2\pi k}{N}+\pi\big)}{\epsilon\big(\frac{2\pi k}{N}\big)}}
\bigg( \sum_{j=1}^d \Big( \sum_{\ell=0}^{n}\frac{c_\ell}{d} \cos\Big(\ell\frac{2\pi k_j}{N}\Big) +  \frac{c_\infty}{d} \cos(\tilde k_j) \Big) \bigg)_+
\notag\\ \stackrel{N \to \infty}{\longrightarrow}& \int_{[0,2\pi]^{2d}}\frac{\text{d}^d k}{(2\pi)^{d}} \, \frac{\text{d}^d\tilde k}{(2\pi)^{d}}\,\sqrt{\frac{\epsilon( k+\pi)}{\epsilon( k)}}
\bigg( \sum_{j=1}^d \Big(\sum_{\ell=0}^{n}\frac{c_\ell}{d} \cos(\ell k_j) + \frac{c_\infty}{d} \cos(\tilde  k_j) \Big) \bigg)_+
\,.
\end{align}
In the same manner, we find that
\begin{align}\label{eq:tilde I_lim}
\lim_{N\rightarrow\infty} \tilde I^d_{c(N)}
&= \int_{[0,2\pi]^{2d}}\frac{\text{d}^dk}{(2\pi)^{d}} \, \frac{\text{d}^d\tilde k}{(2\pi)^{d}}\,\frac{1}{\epsilon(k)}
\bigg( \sum_{j=1}^d \Big(\sum_{\ell=0}^{n}\frac{c_\ell}{d} \cos(\ell k_j) + \frac{c_\infty}{d} \cos(\tilde k_j) \Big) \bigg)_+ \,.
\end{align}
\end{proof}

In this limit, Proposition \ref{l:I_max_m_even} serves to simplify
\[
    I^{u,d}_{c} = \sqrt{1 - u} \, J^d_c,
\]
with $J^d_c$ given in \eqref{eq:Jdc_def}.

\begin{prop}
Let $d\geq1$. For $\eta \geq 0$ set $c(\eta) = (1-\eta,\eta)$. Then 
\begin{align}\label{eq:m->inf}
\liminf_{N\rightarrow\infty} \P_{\beta,\theta,u}\big(E_{0,2 Ne,0}\big) \geq \sup_{\eta
\geq 0} \bigg( 1 - \eta \, h_{{P_{\beta,\theta,u,d}}}\Big( \frac{\colnmb \sqrt{1-u}}{2\sqrt{2}\eta} J^d_{c(\eta),\infty} \Big)
\notag - \frac{\colnmb}{2\beta} \tilde I^d_{c(\eta),\infty} \bigg) \,.
\end{align}
\end{prop}
\begin{proof}
Let $\eta \geq 0$ and set $c(\eta, N) = (1-\eta, \eta, 0, 
\ldots, 0, -1) \in \R^{2N+1}$. Then Proposition 
\ref{c:m>=2_bound} applies, and Proposition \ref{l:I_max_m_even} allows to replace the quantity 
$I^{u,d}_{c(\eta,N)}$ in formula \eqref{eq:m>=2_bound}
by $J^d_{c(\eta,N)}$. Lemma \ref{lem:limit} tells us 
that $\lim_{N \to \infty} J^d_{c(\eta,N)}$ exists and 
provides an expression for it. Continuity of the function
$h_{P_{\beta,\theta,u,d}}$ allows us to take the limit inside. This proves the claim. 
\end{proof}

Numerical values for $m\to\infty$ and $\beta\to\infty$ can be found in Tab. \ref{tab:p_bounds_m>=2}.\\

\begin{table} [!ht]
\centering
{\renewcommand{\arraystretch}{1.2}
\begin{tabular}{ccccccc}
\hline\hline
$d$ & $m=2$ & $m=3$ & $m=4$ & $m=5$ & $m\rightarrow\infty$
\\ \hline
$1$ & $0.047$ & $0$ & $0$ & $0$ & $0$
\\ \hline
$2$ & $0.203$ & $0.142$ & $0.113$ & $0.096$ & $0$
\\ \hline 
$3$ & $0.273$ & $0.251$ & $0.245$ & $0.242$ & $0.237$
\\ \hline
$4$ & $0.306$ & $0.297$ & $0.294$ & $0.294$ & $0.294$
\\ \hline
$5$ & $0.325$ & $0.320$ & $0.319$ & $0.319$ & $0.319$
\\\hline\hline
\end{tabular}
}

\caption{Lower bound of $\P_{\beta,\theta,u}\big(E_{0,me,0}\big)$ for $\beta\to\infty$, $\colnmb=2$, and $u=1/2$ as well as various dimensions $d$ and values of $m\in\mathbb{N}$ including the limit $m\rightarrow\infty$.}
\label{tab:p_bounds_m>=2}
\end{table}

Next, we investigate the range of parameters $\theta$, $u$ and $d$ for 
which long range order can be proved with our method, and compare with 
the results of Ueltschi \cite{Ueltschi_2013}. Equation (5.5) of this
work, stated in our notation, says that long range order holds if 
\begin{equation} \label{e:ueltschi_bound}
    \gamma := \theta \sqrt{1-u} < \sqrt{\frac{2}{J^{d}_{(1,0)} J^{d}_{(0,1)}}}
\end{equation}

To see the connection with the bound we obtain from Theorem \ref{t:main}, we introduce
the function 
\[
b_\gamma(\eta,p) \coloneqq 1-\eta + p \eta - \gamma\sqrt{\frac{p}{2}} J^{d}_{(1-\eta,\eta)}.
\]
Then Theorem \ref{t:main} with the choice $c = (1-\eta, \eta)$, together with 
Proposition \ref{l:I_max_m_even} and Proposition \ref{c:m=1_bound}
gives 
\[
\limsup_{\beta \to \infty} \lim_{|\Lambda|\rightarrow\infty} \bigg( \frac{1}{|\Lambda|} 
\sum_{x\in\Lambda} \P_{\beta,u,\theta}(E_{0,x,0}) \bigg) \geq b_{\rm new}(\gamma) := 
\inf_{p\in[P_{\theta,u,d},1]} \sup_{\eta \in[0,1]} b_\gamma(\eta,p),
\]
and so the parameter $\gamma$ needs to be chosen small enough for the right hand side of 
this to be positive. In comparison, the inequality \eqref{e:ueltschi_bound} comes from the requirement that the right hand side of  
\[
\limsup_{\beta \to \infty} \lim_{|\Lambda|\rightarrow\infty} \bigg( \frac{1}{|\Lambda|} 
\sum_{x\in\Lambda} \P_{\beta,u,\theta}(E_{0,x,0}) \bigg) \geq b_{\text{Ueltschi}}(\gamma)
:= \inf_{p\in[0,1]} \sup_{\eta \in \{ 0,1 \}} b_\gamma(\eta,p)
\]
needs to be positive. Clearly, $b_{\text{new}}(\gamma) \geq b_{\text{Ueltschi}}(\gamma)$ 
for all $\gamma > 0$. While a strict inequality cannot be inferred from the above 
formulae, numerical computations show that it holds. We are thus able to improve the 
range of values $u \in [0,1/2]$ and $\theta \in \N$, $\theta \geq 2$ (in terms of 
$\gamma$) for which long range order can be shown. Some values for  
\[
\gamma^d_{\text{Ueltschi}} = \sup \{ \gamma > 0: b_{\text{Ueltschi}}(\gamma) > 0\}
\qquad \text{and} \quad 
\gamma^d_{\text{new}} = \sup \{ \gamma > 0: b_{\text{new}}(\gamma) > 0\}
\]
are given in Table \ref{tab:gamma_bounds}.

Finally, it is interesting to see how the results obtained via reflection positivity compare to numerical computations of the critical value 
\[
\beta_{\rm crit} = \inf \left\{ \beta > 0: \lim_{|\Lambda|\rightarrow\infty} \bigg( \frac{1}{|\Lambda|} \sum_{x\in\Lambda}  \P_{\beta,u,\theta}(E_{0,x,0}) \bigg) >0 \right\}
\]
for $\beta$ with respect to the appearance of long range order. As above, Theorem \ref{t:main} yields 
\begin{align}
\lim_{|\Lambda|\rightarrow\infty} \bigg( \frac{1}{|\Lambda|} \sum_{x\in\Lambda}  \P_{\beta,u,\theta}(E_{0,x,0}) \bigg) &\geq 
1 - \eta \, h_{0}\Big( \frac{\gamma}{2\sqrt{2}\eta} J^d_{(1-\eta,\eta)} \Big) - \frac{\colnmb}{2\beta} \tilde I^d_{(1-\eta,\eta)},
\end{align}
and thus guarantees long range order whenever $\beta > \tilde \beta_{\rm crit}$, with 
\begin{align*}
    \tilde\beta_\text{crit} \coloneqq \sup_{\eta\in[0,1]} \frac{\frac{\colnmb}{2}\tilde I^d_{(1-\eta,\eta)}}{1 - \eta \, h_{0}\Big( \frac{\gamma}{2\sqrt{2}\eta} J^d_{(1-\eta,\eta)} \Big)} \,.
\end{align*}
 In the case of $\theta=2$ and $d=3$, numerical computations of $\beta_c$ have been carried out \cite{barp2015numerical}, with the results that $\beta_\text{crit}=0.346$ for the case $u=0$, and  $\beta_\text{crit}=0.313$ for the case $u = 1/2$. This compares with  
$\tilde\beta_\text{crit}=1.42$ when $u=0$ and $\tilde\beta_\text{crit}=0.52$ when 
$u=1/2$. The best results are obtained for $u=1/2$, where the two numbers differ only by
a factor of 1.66.
\begin{table} \label{table1}
\centering
{\renewcommand{\arraystretch}{1.2}
\begin{tabular}{cccccc}
\hline\hline
$d$ & $3$ & $4$ & $5$ & $6$ & $7$
\\ \hline
$\gamma^d_\text{Ueltschi}$ & $2.22$ & $2.68$ & $3.00$ & $3.26$ & $3.48$
\\ \hline
$\gamma^d_\text{new}$ & $2.47$ & $3.04$ & $3.46$ & $3.80$ & $4.08$
\\\hline\hline
\end{tabular}
}
\caption{Upper bounds $\gamma^d_\text{Ueltschi}$ and $\gamma^d_\text{new}$ of $\gamma\coloneqq\colnmb\sqrt{1-u}$ provided by reference \cite{Ueltschi_2013} and this article, respectively, where long range order is proven for various dimensions $d$.}
\label{tab:gamma_bounds}
\end{table}

\noindent\textbf{Funding acknowledgement.}  AK's research is funded by the Cusanuswerk.

\appendix

\section{Appendix}

We prove that $\hat\kappa(k,0)\geq0$ for all $k\in\Lambda^*$ based on Eq. (3.24) in \cite{Ueltschi_2013}. Since it involves a physical notation, we first summarize some basics.

For $\theta\in\N$, we set the spin $S\in\frac{1}{2}\N$ such that $\theta=2S+1$. The Hilbert space reads as
\begin{align*}
    \mathcal{H} = \bigotimes_{x\in\Lambda}\mathcal{H}_x\,,
\end{align*}
where each $\mathcal{H}_x$ is a copy of $\mathbb{C}^{2S+1}$ and is spanned by $(|s\rangle)_{s\in \{-S,-S+1,\ldots,S\}}$. The spin operator $S^3_x$ acts on $\mathcal{H}_x$ as $S^3_x|s\rangle_x = s|s\rangle_x$. 
To construct the random loop model, we associate crosses and double bars with operators $T_{x,y},Q_{x,y}$ for $\{x,y\}\in\Ecal_L$ as
\begin{align}\label{eq:T_op_def}
    T_{x,y}|a,b\rangle \coloneqq |b,a\rangle, \quad
Q_{x,y}|a,b\rangle \coloneqq \textbf{1}_{a=b}\sum_{c=-S}^S|c,c\rangle
\end{align}
on $\mathcal{H}_x\otimes\mathcal{H}_y$. This serves to introduce the family of Hamiltonians
\begin{align}\label{eq:H}
    H_u\coloneqq -\sum_{\{x,y\}\in\mathcal{E}} \left( u T_{x,y} + (1-u)Q_{x,y} -1 \right)
\end{align}
for $u\in[0,1]$. Next, we let $A,B:\mathcal{H}\rightarrow\mathcal{H}$ be hermitian and define the partition function
$
    Z\coloneqq \text{Tr}\left[ e^{-\beta H_u} \right]
$
as well as the expectation
\begin{align*}
    \langle A\rangle \coloneqq \frac{1}{Z} \text{Tr}\left[ A\, e^{-\beta H_u} \right] \,,
\end{align*}
the Schwinger function
\begin{align*}
    \langle A;B\rangle(t) \coloneqq  \frac{1}{Z}\,\text{Tr}\left[ A^\dagger e^{-t H_u} B e^{-(\beta-t) H_u} \right] \,,
\end{align*}
and the Duhamel two-point function
\begin{align*}
    (A,B)_\text{Duh}\coloneqq \int_0^\beta\text{d}t\, \langle A;B \rangle(t) \,,
\end{align*}
which is a scalar product. Theorem 3.3 in \cite{Ueltschi_2013} states that
\begin{align*}
    \langle S^3_{x};S^3_{x'}\rangle(t) = \frac{1}{3} S(S+1) \P_{\beta,u,\theta}(E_{x,x',t})
\end{align*}
Translation symmetry in $\Lambda$ gives $\P_{\beta,u,\theta}(E_{x,x',t})=\P_{\beta,u,\theta}(E_{0,x'-x,t})$ and thus simplifies
\begin{align*}
    \langle S^3_{x};S^3_{x'}\rangle(t) = \frac{1}{3} S(S+1) \kappa(x'-x,t)
\end{align*}
along with the special case
\begin{align}\label{eq:mean_kappa}
    \langle S^3_{x} S^3_{x'}\rangle = \langle S^3_{x};S^3_{x'}\rangle(0) =  \frac{1}{3} S(S+1) \kappa(x'-x,0)
\end{align}

Finally, we return to the proof of $\kappa(k,0)\geq0$. The first inequality in Eq. (3.24) in [4] claims
\begin{align*}
\frac{1}{\beta} (A,A)_\text{Duh} \leq \frac{1}{2} \langle A^\dagger A + A A^\dagger \rangle \,.
\end{align*}
Moreover, the Duhamel two-point function fulfills $(A,A)_\text{Duh}\geq0$ as a scalar product, which gives
\begin{align*}
    \frac{1}{2} \langle A^\dagger A + A A^\dagger \rangle\geq 0 \,.
\end{align*}
To conclude the proof, we let $k\in\Lambda^*$ and set $A\coloneqq \sum_{x\in\Lambda} e^{-ik\cdot x} S^3_x$ to give
\begin{align*}
\frac{1}{2}\langle A^\dagger A + A A^\dagger \rangle &= \frac{1}{2} \sum_{x,x'\in\Lambda} \langle S^3_x S^3_{x'}\rangle \left( e^{-ik\cdot x}e^{ik\cdot x'} + e^{ik\cdot x}e^{-ik\cdot x'} \right)
\notag\\&= \frac{S(S+1)}{3} \sum_{x,x'\in\Lambda} \kappa(x'-x,0) \cos \big(k\cdot (x'-x)\big)
\notag\\&= |\Lambda| \frac{S(S+1)}{3} \sum_{x\in\Lambda} \kappa(x,0) \cos (k\cdot x)
\notag\\&= |\Lambda| \frac{S(S+1)}{3} \hat\kappa(k,0)
\end{align*}
using Eq. \eqref{eq:mean_kappa} and translation invariance, which yields $\hat\kappa(k,0) \geq0$.

\bibliographystyle{abbrv}
\bibliography{bib}

\end{document}